\documentclass[prd,11pt,tightenlines,nofootinbib,superscriptaddress]{revtex4}
\usepackage{amsfonts,amssymb,amsthm,bbm,amsmath}
\usepackage{verbatim}
\usepackage{graphicx}

\newcommand{\C}{{\mathbb C}}

\newcommand{\one}{\mathbbm{1}}

\newcommand{\cC}{{\mathcal C}}

\newcommand{\SU}{\mathrm{SU}}

\newcommand{\SL}{\mathrm{SL}}

\def\T{T^{\Gamma}}
\def\G{{\cal G}_{\Gamma}}

\newcommand{\be}{\begin{equation}}
\newcommand{\ee}{\end{equation}}
\newcommand{\beq}{\begin{eqnarray}}
\newcommand{\eeq}{\end{eqnarray}}
\newcommand{\bea}{\begin{eqnarray}}
\newcommand{\eea}{\end{eqnarray}}

\newcommand{\mat} [2] {\left ( \begin{array}{#1}#2\end{array} \right ) }

\newcommand{\bra}{\langle}
\newcommand{\ket}{\rangle}

\newcommand{\ra}{\rangle}

\newcommand{\sgn}{\mathrm{sgn}}

\newcommand{\rd}{\mathrm{d}}

\newcommand{\bpm}{\begin{pmatrix}}
\newcommand{\epm}{\end{pmatrix}}
\newcommand{\bvm}{\begin{vmatrix}}
\newcommand{\evm}{\end{vmatrix}}

\newtheorem{theorem}{Theorem}[section]
\newtheorem{lemma}[theorem]{Lemma}
\newtheorem{proposition}[theorem]{Proposition}
\newtheorem{corollary}[theorem]{Corollary}
\newtheorem{definition}[theorem]{Definition}

\begin{document}

\title{On the exact evaluation of spin networks}

\author{{\bf Laurent Freidel}\email{lfreidel@perimeterinsititute.ca}}
\author{{\bf Jeff Hnybida}\email{jhnybida@perimeterinsititute.ca}}
\affiliation{ Perimeter Institute for Theoretical Physics,
Waterloo, N2L-2Y5, Ontario, Canada.}

\begin{abstract}
We introduce a fully coherent spin network amplitude whose expansion generates all SU$(2)$ spin networks 
associated with a given graph.
We then  give an explicit evaluation of this  amplitude for an arbitrary graph.
We show how this coherent amplitude can be obtained from the specialization of a generating functional obtained by 
the contraction of parametrized intertwiners \`a la Schwinger. We finally give the explicit evaluation of this generating functional for arbitrary graphs.
\end{abstract}

\maketitle


\section{Introduction}
Spin networks consist of graphs labeled by representations of SU$(2)$. They play a fundamental role in quantum gravity in two respects.
First, they arise as a basis of states for 4 dimensional quantum gravity formulated in terms of Ashtekar-Barbero variables \cite{AL}.
Second, in spin foam models the evaluations of spin networks appear in the expression of the quantum transition amplitude between boundary spin network states. 
In 3 dimensions the amplitude associated with a tetrahedron is the 6j symbol or the spin network evaluation of a tetrahedral graph \cite{PR}.
Similarly, in 4 dimensions the amplitude associated with a 4-simplex is given by a product of two 15j-symbols with labels being related via the Immirzi 
parameter \cite{EPRL, FK}.
These results have triggered an extensive study of the properties of these symbols especially concerning their asymptotic properties.
It has been known for a while that the 6j symbol is asymptotically related to the cosine of the Regge action, and
moreover  it has been shown more recently that these results extend to the 15j symbol \cite{BarrettAs1,BarrettAs2,FC1,FC2}.
In order to understand the asymptotic property of the spin network evaluation,
a key insight was to express these amplitudes in terms of coherent states and express the amplitudes as a functional 
of spinor variables.
Furthermore, the  coherent state representation has been found to admit a corresponding geometrical interpretation in terms of 
polyhedral \cite{poly} and twisted geometries \cite{twistedgeo} where the spin labels represent the area of a polygonal face.
It was also recognized that further simplifications of the amplitudes and other structures involving the coherent intertwiners could be achieved by summing
the amplitudes over the spins with certain weight while keeping the total spin associated with each of the vertices fixed.
In this case the amplitudes were found to exhibit an extra U$(N)$ symmetry that renders certain computations extremely efficient \cite{UN1,UN2,UN3}.
What we propose  here is to go one step further and consider fully coherent spin network amplitudes 
obtained by summing over all spins with a specific weight.
Such a proposal has also been developed recently by Livine and Dupuis \cite{LD} in order to write spin foam models more efficiently.

What we show is that by carefully choosing the weight of the spin network amplitudes we can compute them {\it exactly} for an arbitrary graph. This is our main result.
These fully coherent amplitudes contain as their expansion coefficients the spin network evaluations  for arbitrary spins and therefore 
 can be understood as {\it a generating functional } for all spin networks evaluations.
 Similar generating functionals have been studied before, first by Schwinger \cite{Schwinger} and then further developed by Bargmann \cite{Bargmann} (see also \cite{Labarthe}).
Much later, an explicit evaluation of the generating functional for the {\it chromatic} (or Penrose \cite{Penrose}) evaluation of a spin network on a planar trivalent graph was given by Westbury \cite{Westbury}.
More  recently,  Garoufalidis et al. \cite{Garoufalidis} extended the evaluation of the chromatic  generating functional  to non planar graphs and  Costantino and Marche \cite{Costantino} to the case where holonomies along the edges are present.
In our case we focus on a slightly different generating functional that does not generate the chromatic evaluation but rather the usual spin network  evaluation (i.e. the one obtained by the contraction of intertwiners).  This evaluation differs from the chromatic one by an overall sign \cite{BarrettPR} and is much simpler in the non planar case.
We present new techniques that allow this  generating functional to be represented as a Gaussian integral and finally as the reciprocal of a polynomial.  Furthermore, these results are valid for graphs of full generality such as those which are non-planar or of higher valency.

\section{Coherent evaluation of the vertex and general graphs}
One of the key recent developments concerning spin foam amplitudes has been the ability to express them in terms of  SU$(2)$ coherent states.
In the following we denote the coherent states and their contragradient version by
\be
|z\ra\equiv \mat{c}{\alpha\\ \beta}, \qquad
| z] \equiv \mat{c}{-\bar{\beta}\\\bar{\alpha}}.
\ee
The bracket between these two spinors $[z_{1}|z_{2}\ket= \alpha_{1}\beta_{2}-\alpha_{2}\beta_{1}$ is purely holomorphic
and anti-symmetric with respect to the exchange of $z_{1}$ with $z_2$ .
We will also denote the conjugate states by $\bra z| =(\bar{\alpha},\bar{\beta})$ and $[z|=(-\beta,\alpha)$.

The vertex amplitude for SU$(2)$ BF theory, expressed in terms of the SU$(2)$ coherent states,
depends on $10$ spins $j_{ij}$ and depends holomorphically on $20$ spinors $|z_{ij}\ket \neq |z_{ji}\ket$; it is based on a 4-simplex graph and is given by:
\be
A_{4S}(j_{ij}, z_{ij}) =\int_{\mathrm{SU(2)^{5}}}
\prod_{i} \rd g_{i}   
 \prod_{i<j} [z_{ij}|g_{i}^{-1}g_{j}|z_{ji}\ket^{2j_{ij}}.
\ee
One of the key advantages of expressing the vertex amplitude in terms of coherent states is the ease with which to compute  
its asymptotic properties \cite{BarrettAs1, BarrettAs2, FC2}.
However, even if the asymptotic property of this amplitude is known, we do not know how to
 to compute it explicitly.
What we are going to show is that by resumming these amplitudes  in terms of a fully coherent amplitudes
where $j_{ij}$ is summed over, we can get an {\it exact} expression for the vertex amplitude.

Let us denote $J_{i} \equiv \sum_{j|j\neq i } j_{ij}$ and define the following vertex amplitude
which now depends only on the spinors:
\be\label{newvertex}
{\cal{A}}_{4S}(z_{ij} ) \equiv \sum_{j_{ij}} \frac{\prod_{i} (J_{i}+1)!}{\prod_{i<j}(2j_{ij})!} A_{4S}(j_{ij}, z_{ij}).
\ee
Since $ |A_{4S}(j_{ij}, z_{ij})| \leq \prod_{i<j}| [z_{ij}|z_{ji}\ket|^{2j_{ij}}$ it can be easily seen that such a series admits a non-zero radius of convergence.
This amplitude can be thought of as a generating functional for the vertex amplitude $A_{4S}(j_{ij}, z_{ij})$, where the magnitudes of the $\bra z_{ij}|z_{ij}\ket$ determine which spin $j_{ij}$ the amplitude is peaked on.
Let us finally note that the particular set of coefficients we use in order to sum 
the coherent state amplitudes is motivated by the U$(N)$ perspective \cite{UN1,UN2}.  That is, if we denote by $||j_{i},z_{i}\ket $ the SU$(2)$ coherent intertwiners and by $|J, z_{i})$ the SU$(N)$ coherent states we have the relation
\be\label{un}
\frac{|J,z_{i})}{\sqrt{J!}} = \sum_{\sum_{i}j_{i}=J} \sqrt{\frac{{(J+1)!}}{\prod_{i} (2j_{i})!}} ||j_{i},z_{i}\ket
\ee 
where 
$
||j_{i},z_{i}\ket =\int \rd g \,\otimes_{i} \, \left(g |z_{i}\ket \right)^{ 2j_{i}},
$
is the coherent intertwiner \cite{LS}.
The idea to use generalized coherent states which include the sum over all spins in a way compatible with the U$(N)$ symmetry, has been already proposed in 
\cite{UN3} in order to treat all simplicity constraints arising in the spin foam formulation of gravity on the same footing.
One ambiguity concerns the choice of the measure factor used to perform the summation over the total spin $J$.
The choice to sum states as $\sum_{J} |J,z_{i})/{\sqrt{J!}}$ differs form the one taken in \cite{UN3}, but is ultimately justified for us by the fact that the spin network amplitude can be exactly evaluated.

The definition of the generally coherent amplitude is not limited to the 4-simplex, it can be extended to any spin network:
More generally, let $\Gamma$ be an oriented graph with  edges denoted by $e$ and vertices  by $v$.
We assign two spinors $z_{e}, z_{e^{-1}}$ to each oriented edge $e$, one for $e$ and one for the reverse oriented edge $e^{-1}$.
We also assign spins $j_{e}= j_{e^{-1}}$ to every edge and define
$$J_{v} \equiv\sum_{e:s_{e}=v} j_{e} + \sum_{e: t_{e}=v} j_{e}$$
where $s_{e}$ (resp. $t_{e}$) is the starting (resp. terminal vertex) of the edge $e$.

Given this data we define a functional depending on $j_{e}$ and holomorphically on all $z_{e}$  given by:
\be\label{def1}
A_{\Gamma}(j_{e},z_{e}) 
\equiv \int \prod_{v\in V_{\Gamma}} \rd g_{v} \prod_{e \in E_{\Gamma}}[z_{e}|g_{s_e}g_{t_e}^{-1}|z_{e^{-1}}\ket^{2j_{e}}
\ee
where we define $E_{\Gamma}$ to be the set 
of  edges of $\Gamma$ and  $V_{\Gamma}$  the set of vertices. 
Finally, we introduce the following amplitude depending purely on the spinors
\bea\label{def2}
{\cal A}_{\Gamma}(z_{e})
&\equiv& \sum_{j_{e}} \frac{ \prod_{v\in V_{\Gamma}}(J_{v}+1)!}{\prod_{e\in E_{\Gamma} }(2j_{e})!}A_{\Gamma}(j_{e},z_{e}), \\
&=& \sum_{j_{e}} \prod_{v\in V_{\Gamma}}(J_{v}+1)! \int \prod_{i\in V_{\Gamma}} \rd g_{i} \prod_{e \in E_{\Gamma}} \frac{[z_{e}|g_{s_e}g_{t_e}^{-1}|z_{e^{-1}}\ket^{2j_{e}}}{(2j_{e})!}. \label{def3}
\eea
The main motivation for this definition comes from the fact that it can be explicitly evaluated, and this follows from the fact that this amplitude can be expressed as a Gaussian integral.
\begin{lemma}
The fully coherent amplitude can be evaluated as a Gaussian integral
\bea
{\cal A}_{\Gamma}(z_{e}) 
&=& \int_{\C^{2|V_{\gamma}|}} \prod_{i \in V_{\Gamma}}\rd\mu(\alpha_{i}) 
\exp\left(-{\sum_{i,j \in V_{\Gamma}}\bra \alpha_{i} | X_{ij} |\alpha_{j} \ket}\right) 
\eea
where $\rd\mu(\alpha) \equiv e^{-\bra\alpha |\alpha \ket } \rd^{4} \alpha / \pi^{2}$; 
and $X_{ij}$ is a 2 by 2 matrix which vanishes if there is no edge between $i$ and $j$.
If $(ij)=e$ is an edge of $\Gamma$,  $X_{ij}$ is given by
\be \label{eqn_matrix_X}
X_{ij} = \sum_{e|s_{e}=i, t_{e}=j} |z_{e}\ket [z_{e^{-1}}| - \sum_{e|t_{e}=i, s_{e}=j} |z_{e^{-1}}\ket [{z}_{e}|.
\ee
 This Gaussian integral can be evaluated giving
\be \label{det1}
{\cal A}_{\Gamma}(z_{e}) = \frac{1}{\det(1+X(z_{e}))}.
\ee
\end{lemma}
\begin{proof}
Given a group element $g_{i}\in \SU(2)$ we can construct a unit spinor  $|\alpha_{i}\ket \equiv g_{i}^{-1}|0\ket$ where $|0\ket = (1 \: 0)^T$.
Using the the decomposition of the identity $\one=|0\ket\bra0| + |0][0|$,  we can express the group product
as 
\bea
g_{i}^{-1} g_{j} & = & g_{i}^{-1}(|0\ket\bra0| + |0][0|) g_{j} = |\alpha_{i}\ket\bra \alpha_{j}| + |\alpha_{i}][\alpha_{j}|
\eea
where we used the decomposition of the identity $\one=|0\ket\bra0| + |0][0|$.  
On the other hand, given a spinor $|\alpha\ket$ we can construct a group element $g(\alpha) \equiv 
|0\ket\bra \alpha| + |0][\alpha|$, for which $ g^{\dagger}(\alpha) g(\alpha) = \bra \alpha | \alpha \ket$.
Any function of $ |\alpha\ket $ and its conjugate can be viewed as a function
$F(g(|\alpha\ket))$ of this element and therefore can be viewed, when restricted to unit spinors as a function  on SU$(2)$.
Lets now suppose that $F(g(|\alpha\ket))$ is homogeneous of degree $2J$ in $|\alpha \ket$, i.e.
$F(g(\lambda \alpha)) = \lambda^{2J}F( g(\alpha))$  for $\lambda >0$.  Then we can express the group integration as a Gaussian integral over spinors
\be \label{eqn_homo_integral}
 (J+1)! \int_{SU(2)}\rd g  F(g) = \frac{ 1 }{\pi^{2}} \int_{\C^{2}} \rd^{4} \alpha\, e^{-\bra \alpha | \alpha \ket} F(g(\alpha))
\ee
where $\rd g$ is the normalized Haar measure (for proof see the appendix).  Therefore ${\cal A}_{\Gamma}(z_{e}) $ can be written as a 
Gaussian integral
\bea
{\cal A}_{\Gamma}(z_{e}) 
&=& \int_{\C^{2|V_{\gamma}|}} \prod_{i \in V_{\Gamma}}\rd\mu(\alpha_{i}) 
\exp\left({\sum_{e \in E_{\Gamma}} [z_{e}| \left( |\alpha_{s(e)}\ket\bra \alpha_{t(e)}| + |\alpha_{s(e)}][\alpha_{t(e)}| \right) |z_{e^{-1}}\ket} \right) 
\eea
where $\rd\mu(\alpha) \equiv e^{-\bra\alpha |\alpha \ket } \rd^{4} \alpha / \pi^{2}$. 

Using the relation $[\alpha |w\ket [z|\beta] = - \bra \beta | z\ket [w|\alpha\ket$ we can write the integrand as 
$\exp\left(-{\sum_{i,j \in V_{\Gamma}}\bra \alpha_{i} | X_{ij} |\alpha_{j} \ket}\right) $
where the 2 by 2 matrix $X_{ij}$ is given by
\be 
X_{ij} = \sum_{e|s_{e}=i, t_{e}=j} |z_{e}\ket [z_{e^{-1}}| - \sum_{e|t_{e}=i, s_{e}=j} |z_{e^{-1}}\ket [{z}_{e}| 
\ee
and $X_{ij}$ vanishes if there is no edge between $i$ and $j$.  This Gaussian integral can be easily evaluated giving the determinant formula (\ref{det1}).
\end{proof}
We  now want to  evaluate this determinant explicitly.  To do this we require the following definitions.

\begin{definition} A loop of $\Gamma$ is a set of edges $l= e_{1}, \cdots e_{n}$ such that $ t_{e_{i}}= s_{e_{i+1}}$ and $ t_{e_{n}}= s_{e_{1}}$.  
A simple loop of $\Gamma$ is a loop in which $e_{i}\neq e_{j}$ for $i\neq j$, that is each edge enters at most once.
A non trivial cycle $c= (e_{1}, \cdots e_{n} )$ of $\Gamma$ is a simple  loop of $\Gamma$ in which  $s_{e_i} \neq s_{e_j}$ for $i \neq j$, 
i.e. it is a simple loop in which  each vertex is  traversed at most once.
A disjoint cycle union of $\Gamma$ is a collection $C=\{c_{1},\cdots, c_{k}\}$ of non trivial cycles of $\Gamma$ which are pairwise disjoint (i.e. do not have any common edges or vertices). Given a non trivial cycle $c= (e_{1}, \cdots ,e_{n})$ we define the quantity
\be
A_{c}(z_{e}) \equiv -(-1)^{|e|} [\tilde{z}_{e_{1}} | z_{e_{2}}\ket [\tilde{z}_{e_{2}}|z_{e_{3}}\ket \cdots  [\tilde{z}_{e_{n}} | z_{e_{1}}\ket
\ee
where 
$|e|$ is the number of edges of $c$ whose orientation agrees with the chosen orientation of  $\Gamma$,
and $\tilde{z}_{e}\equiv z_{e^{-1}}$.  Finally, given a disjoint cycle union $C=\{c_{1}\cdots c_{k}\}$ we define
\be \label{eqn_cycle_union_amp}
A_{C}(z_{e}) = A_{c_{1}}(z_{e})\cdots A_{c_{k}}(z_{e}).
\ee
\end{definition}
With these definitions we present the final expression for the vertex amplitude in the following theorem.
\begin{theorem} \label{thm_amp}
\be
{\cal A}_{\Gamma}(z_{e}) = \frac1{\left(1 + \sum_{C} A_{C}(z_{e})\right)^{2}}
\ee
where the sum is over all disjoint cycle unions $C$ of $\Gamma$.
\end{theorem}
The proof of this result is detailed in the appendix, and is due to the following special property of the matrix $X$.
\begin{proposition}
The  Matrix  $X$ defined in Eq. (\ref{eqn_matrix_X}) is what we call a scalar loop matrix.  That is for any collection of indices $L = (i_{1},\cdots, i_{n})$ of $\{1,2,...,n\}$ where $n$ is the size of $X$ the quantity
\be
\frac{1}{2}\left(X_{i_{1}i_{2}}X_{i_{2}i_{3}}\cdots X_{i_{n}i_{1}} +  X_{i_{1}i_{n}}X_{i_{n}i_{n-1}}\cdots X_{i_{2}i_{1}} \right) = X_{L} \one
\ee
is proportional to the identity. 
\end{proposition}
This property allows us to prove the following lemma from which the theorem follows:
\begin{lemma}
If $X$ is a  $n\times n$ scalar loop matrix composed of 2 by 2 block matrices then 
\be
\det(X) = \left(\sum_{C} \sgn(C) X_{i_1} \cdots X_{i_k} \right)^{2},
\ee
where the sum is over all collections of pairwise disjoint cycles $C = (i_1, \cdots, i_k)$ of $\{1,\cdots, n\}$ which cover $\{1,2,...,n\}$, and $\sgn(C)$ is the signature of $C$ viewed as a permutation of $(1,\cdots, n)$.
\end{lemma}
Evaluating this sum leads to our main theorem.

\subsection{Illustration}
Let us first illustrate this theorem on one of the simplest graphs: the theta graph $\Theta_{n}$.  This graph consists of two vertices with $n$ edges running between them.  
The amplitude for this graph depends on $2n$ spinors denoted $z_{i}$ for the spinors attached to the first vertex and $w_{i}$ for the ones attached to the second vertex.  
The orientation of all the edges is directed from $z_{i}$ to $w_{i}$ where $i=1,\cdots, n$ labels the edges of $\Theta_{n}$.
For this graph the only cycles which have non-zero amplitudes are of length 2.  Further, since there are only two vertices, each disjoint cycle union
consists of a single nontrivial cycle. The amplitude associated to such a cycle going along the edge $i$ and then $j$ 
is given by
\be
A_{ij} = [w_{i}|w_{j}\ket[z_{j}|z_{i}\ket
\ee
Therefore, from our general formula we have
\be\label{theta}
{\cal A}_{\Theta_{n}}(z_{i},w_{i}) = \left(1 + \sum_{i<j}[w_{i}|w_{j}\ket[z_{j}|z_{i}\ket \right)^{-2}.
\ee

We now illustrate the theorem for cases of the 3-simplex and the 4-simplex. 
In a $n$-simplex there is exactly one oriented edge for any pair of vertices $e=[ij]$ and so we can label cycles by sequences of vertices.  We choose the 
orientation of the simplex to be such that positively oriented edges are given by $e=[ij]$ for $i<j$.
Associated to the oriented edge $e=[ij]$ we assign the spinors $$z_{e} \equiv z^{i}_{j},\qquad \tilde{z}_{e}=
z_{e^{-1}} \equiv z^{j}_{i}.$$ 
Given a non trivial cycle $(1,2, \dots, p)$ of a $n$-simplex we define its amplitude by
\be
A_{12\cdots p}\equiv   [z^{1}_{p}|z_{2}^{1}\ket[z_{1}^{2}|z_{3}^{2}\ket \cdots [z_{p-1}^{p}|z_{1}^{p}\ket.
\ee
For the 3-simplex we have four non-trivial cycles of length $3$ and three non-trivial cycles of length $4$.  Since each of these cycles share a vertex or edge with every other, the only disjoint cycle unions are those which contain one non-trivial cycle.  Therefore, after taking into account the sign convention the 3-simplex amplitude is given by
\be
{\cal A}_{3S}= \bigg(1 - A_{123} - A_{124} - A_{134} - A_{234} + A_{1234} - A_{1243} - A_{1324} \bigg)^{-2}.
\ee
The sign in front of $A_{123}$ is determined in the following way.  First, there is one $-1$ which comes from the cycle union having one non trivial cycle and two $-1$ because the non trivial cycle $(1,2,3)$ contains the two edges $12$ and $23$ which have a positive orientation.  Thus the sign is negative.

For the 4-simplex we have ten 3-cycles, fifteen 4 cycles, and twelve 5 cycles and again the disjoint cycle unions consist of only single cycles.
We define the 3-cycle amplitude to be
\be
A_{3}\equiv A_{123} + A_{124}+ A_{134} + A_{234} + A_{125}  +A_{135}+ A_{345} + A_{145} + A_{245} + A_{345},
\ee
the 4-cycle amplitude to be
\be
A_{4}\equiv \hat{A}_{1234} + \hat{A}_{1235} + \hat{A}_{1245} + \hat{A}_{1345} + \hat{A}_{2345},
\quad \mathrm{
with
}\quad 
\hat{A}_{1234}= A_{1234}- A_{1324}- A_{1243}.
\ee
and the 5-cycle amplitude to be
\bea\nonumber
A_{5} &=&  A_{12345} - A_{12435} - A_{23541}- A_{34152} - A_{45213}- A_{51324}\\
& &
- A_{12453} - A_{23514} -A_{34125} - A_{45231}- A_{51342} - A_{13524}.
\eea
Finally, the 4-simplex amplitude is given by
\be
{\cal A}_{4S} = ( 1 - A_{3} + A_{4} - A_{5})^{-2}.
\ee

\section{ Intertwiners and The vertex Amplitude}

The goal of this section is to understand more deeply the 
relationship between the coherent evaluation of  3 and 4-valent graphs like the 3 and 4-simplex and the usual 
evaluation of spin network.

In order to express the coherent evaluation ${\cal A}_{3S}$ and ${\cal A}_{4S}$,  in terms of the 6j  and 15j symbols respectively,
we need to know the relationship between the coherent intertwiner and the normalised 3j symbol.
This relationship is  well known for 3-valent intertwiners \cite{Bargmann,holomorph,UN2,livine-bonzom}, however we will give an independent 
and elegant derivation that will allow us to understand this relationship in the unknown 4-valent case ( for an exception see \cite{UN2}).

\subsection{The n-valent intertwiner}

It is well-known \cite{ Schwinger, Bargmann} that the spin $j$ 
representation can be understood in terms of holomorphic functions on spinor space $\C^{2}$ which are homogeneous
 of degree $2j$.  In this formulation a holomorphic and orthonormal basis corresponding to the diagonalisation of $J_{3}$ is given by
\be
   e^{j}_{m}(z) = \frac{\alpha^{j+m}\beta^{j-m}}{\sqrt{(j+m)!(j-m)!}} 
\ee
where $(\alpha, \beta)$ are the components of the spinor $|z\ket$.
This basis is orthonormal with respect to the Gaussian measure
\be
  d\mu(z) = \frac{1}{\pi^2} e^{-\bra z | z \ket}  \rd^{4}z
\ee
and $ \rd^{4}z$ is the Lebesgue measure on $\C^2$.  In fact these basis elements are the bracket between the usual states and the coherent states
$$e^{j}_{m}(z) = \bra j,m | z\ket $$.
In this representation it is straightforward to construct a basis of  $n$-valent intertwiners, i.e. functions of $z_{1},\cdots ,z_{n}$ which are invariant under $\SL(2,\C)$ and homogeneous of degree $2j_{i}$ in $z_{i}$.
A complete basis of these intertwiners is labeled by  $n(n-1)/2$ integers $[k]\equiv (k_{ij})_{i\neq j = 1,\cdots, n}$ with $k_{ij}=k_{ji}$ and given by
\be\label{C}
C_{[k]}^{(n)}(z_{i}) \equiv (-1)^{s_{n}} \prod_{i<j}\frac{ [z_{i}|z_{j}\ket^{k_{ij}}}{k_{ij}!}.
\ee
where the sign factor $s_{n}$ is chosen for convenience\footnote{For instance in the trivalent case we take $s_{3}=k_{31}$ so that the ordering correspond to the cyclic ordering with
$z_{12},z_{23},z_{31}$ instead of $z_{12},z_{23},z_{13}$.}.
By homogeneity the integers $[k]$ must satisfy the conditions
\be\label{kj}
\sum_{j\neq i} k_{ij} =2j_{i}
\ee
and when these conditions are satisfied we write $[k]\in K_{j}$.

We now would like to understand the relationship between this basis of intertwiners and the coherent intertwiners, and in particular the scalar product between these states.
In order to investigate this, let us introduce the normalised intertwiner basis
\be\label{Chat}
\widehat{C}_{[k]}^{(n)}(z_{i}) \equiv  \frac{ \prod_{i<j} [z_{i}|z_{j}\ket^{k_{ij}}}{\sqrt{ (J+1)! \prod_{i<j} k_{ij}!}} 
= \sqrt{\frac{ \prod_{i<j} k_{ij}!}{(J+1)!}} C_{[k]}^{(n)}. 
\ee

Intuitively, the theta graph consists of two $n$-valent intertwiners with pairs of legs identified.  Indeed, expanding the theta graph amplitude (\ref{theta}) in a power series yields an expression in terms of these intertwiners
\bea
{\cal A}_{\Theta_{n}}(z_{i},w_{i}) &=& \sum_{J} (-1)^J (J+1) \left(\sum_{i<j}  [w_{i}|w_{j}\ket[z_{j}|z_{i}\ket\right)^{J} \\
&=& \sum_{[k]}     {(J+1)!} \frac{
\prod_{i<j}  [w_{i}|w_{j}\ket^{k_{ij}} [z_{i}|z_{j}\ket^{k_{ij}}
}{\prod_{i<j} k_{ij}!}
\\ 
&=& \sum_{j_{i}} \left[(J+1)! \right]^{2}  \sum_{[k] \in K_{j}}\widehat{C}_{[k]}^{(n)}(z_{i})\widehat{C}_{[k]}^{(n)}(w_{i}).
\eea
This shows that ${\cal A}_{\Theta_{n}}(z_{i},w_{i})$ is a generating functional for the $n$-valent intertwiners.
Given the definition (\ref{def2}) of the amplitude ${\cal A}_{\Theta_{n}}(z_{i},w_{i})$ in terms of coherent intertwiners, this implies that 
\be\label{CC}
  \sum_{[k] \in K_{j}}\widehat{C}_{[k]}^{(n)}(z_{i})\widehat{C}_{[k]}^{(n)}(w_{i}) =\int \rd g  \prod_{i}\frac{[z_{i}|g|w_{i}\ket^{2j_{i}} }{(2j_{i})!}.
\ee
This shows that the relation between  the coherent intertwiner $\|j_{i}, z_{i} \ket$ and the  normalised $n$-valent intertwiner $\widehat{C}_{[k]}$ is given by
\be\label{conor}
\frac{\|j_{i}, z_{i} \ket}{\sqrt{ \prod_{i} (2j_{i})!}} = \sum_{[k] \in K_{j}} \left|\widehat{C}_{[k]}^{(n)}\right\ket\widehat{C}_{[k]}^{(n)}(z_{i})
\ee
where we have introduce the state $ \left\bra \left. \widehat{C}_{[k]}^{(n)}\right| z_{i}\right\ket \equiv \widehat{C}_{[k]}^{(n)}(z_{i})$.

We now have to understand the normalization properties of  $\widehat{C}_{[k]}^{(n)}$.
In order to do so, it is convenient to introduce another generating functional defined by
\be\label{ACC}
\widehat{\cal A}_{\Theta_{n}}(z_{i},w_{i}) \equiv  \sum_{[k]}   \widehat{C}_{[k]}^{(n)}(z_{i})\widehat{C}_{[k]}^{(n)}(w_{i}).
\ee
The remarkable fact about this generating functional, which follows from (\ref{CC}), is that it can be written as the evaluation of the following integral
\be
\widehat{\cal A}_{\Theta_{n}}(z_{i},w_{i}) =\int_{\SU(2)} \rd g \, e^{\sum_{i}[z_{i}|g|w_{i}\ket} \, .
\ee
We can now compute 
\bea
\int \prod_{i}\rd\mu(w_{i}) \left|\widehat{\cal A}_{\Theta_{n}}(z_{i},w_{i})\right|^{2} &=& \int \rd g \rd h 
\int \prod_{i}\rd\mu(w_{i}) e^{\sum_{i}[z_{i}|g|w_{i}\ket + \sum_{i}\bra w_{i}|h^{-1}|z_{i} ]}\\
&=& \int \rd g \rd h \,e^{\sum_{i}[z_{i}|gh^{-1}|z_{i} ]} 
= \widehat{\cal A}_{\Theta_{n}}(z_{i}, \check{z}_{i})
\eea
where $ |\check{z}_{i}\ket \equiv |z_{i}]$ and in the second line we performed the Gaussian integral.

Using (\ref{ACC}) to write this equality in terms of the intertwiner basis we get
\be
\sum_{[k],[k']} \widehat{C}_{[k']}^{(n)}(z) \left\bra \widehat{C}_{[k']}^{(n)} \right|\left. \widehat{C}_{[k]}^{(n)} \right\ket  \widehat{C}_{[k]}^{(n)}(\check{z}_{i})
=   \sum_{[k]}   \widehat{C}_{[k]}^{(n)}(z_{i})\widehat{C}_{[k]}^{(n)}(\check{z}_{i})
\ee
where we have used that ${C}_{[k]}^{(n)}(\check{z}_{i})$ is the complex conjugate\footnote{ We have 
$[\check{w}|\check{z}\ket=-\bra w | z ]= \bra z|w]= \overline{[w|z\ket}.$} of $C_{[k]}^{(n)}({z}_{i})$.
This shows that the combination 
\be\label{proj}
P _{j} \equiv  \sum_{[k]\in K_{j} }   \left| \widehat{C}_{[k]}^{(n)}\right\ket \left\bra \widehat{C}_{[k]}^{(n)}\right| 
\ee
is a projector onto the space of SU$(2)$ intertwiners of spin $j_{i}$.

In the case $n=3$ there is only one intertwiner.  Indeed, given $[k]=(k_{12},k_{23},k_{31})$ the  homogeneity restriction requires $2j_{1}=k_{12} + k_{13}$ which can be easily solved by
\be
k_{ij} = J - 2j_{i},\qquad J\equiv j_{1}+j_{2}+j_{3}.
\ee
In this case the fact that $P_{j}$ is a projector implies that $C_{[k]}^{(3)}$ form an orthonormal basis, 
$ \left\bra \widehat{C}_{[k]}^{(3)} | \widehat{C}_{[k']}^{(3)} \right\ket = \delta_{[k],[k']}$.
In other word we can write 
\be
\widehat{C}_{[k]}^{(3)}(z_{i}) = \sum_{m_{i}} 
\left(
\begin{array}{ccc}
j_{1} & j_{2} & j_{3} \\
m_{1} & m_{2} & m_{3}
\end{array}
\right) e^{j_{1}}_{m_{1}}(z_{1})e^{j_{2}}_{m_{2}}(z_{2})e^{j_{3}}_{m_{3}}(z_{3})
\ee
where the coefficients are the Wigner 3$j$ symbols.

Using the relationship (\ref{conor}) between the normalised and coherent intertwiners and the definition (\ref{def3}) of the amplitude in terms of coherent intertwiners we can evaluate the 3-simplex amplitude in terms of the 6$j$ symbol as
\be
{\cal A}_{3S}(z_{j}^{i}) =\sum_{j_{ij}} \prod_{i} (J_{i}+1)! (-1)^{s}\prod_{i} \widehat{C}_{j_{ij}}(z^{i}_{j})  \left\{ 6j \right\}.
\ee
Here $s=j_{12}+j_{13}$ and this signs comes from the fact that the oriented graph for the 6$j$ symbol differs from the generic orientation we have chosen by a change of order of the edge $12$ and $23$ (see e.g. \cite{LJ} for the definition of the 6$j$).
Note that it is 
also interesting to consider the amplitude
\be
\widehat{\cal A}_{3S}(z_{j}^{i}) \equiv \sum_{j_{ij}}  \prod_{i} \widehat{C}_{j_{ij}}(z^{i}_{j})  \left\{ 6j \right\}=
\int \prod_{i} \rd g_{i} e^{\sum_{i<j} [z^{i}_{j}|g_{i}g_{j}^{-1}|z_{i}^{j}\ket}
\ee
although this amplitude cannot be evaluated exactly, unlike $ \cal A$.  This amplitude does however possess interesting asymptotic properties.
\section{Generating Functionals}

We would like now to provide a  direct evaluation of the scalar product between two intertwiners.
In order to do so 
we introduce the following generating functional which depends holomorphically on $n$ spinors $|z_{i}\ket$ and $n(n-1)/2$ complex numbers
$\tau_{ij} =-\tau_{ji}$
\be\label{defC}
 {\cal C}_{\tau_{ij}}(z_{i})  \equiv e^{\sum_{i<j} \tau_{ij}[z_{i}|z_{j}\ket } = \sum_{[k] } 
 \prod_{i<j} \tau_{ij}^{k_{ij}} C_{[k]}(z_{i}).
\ee
This functional was first consider by Schwinger \cite{Schwinger}.
We now compute the scalar product between two such intertwiners
\bea\label{CC2}
\left\bra {\cal C}_{\tau_{ij}} | {\cal C}_{\tau_{ij}} \right\ket & = & \int \prod_{i}\rd\mu(z_{i})   \left|{\cal C}_{\tau_{ij}}(z_{i})\right|^{2}\\
& = &  \int \prod_{i}\rd\mu(z_{i}) e^{\sum_{i<j} \tau_{ij}[z_{i}|z_{j}\ket + \bar{\tau}_{ij} \bra z_{j}|z_{i}]}.
\eea
If we denote by $\alpha_{i}\in \C$ and $\beta_{i} \in \C$ the two components of the spinor $z_{i}$, 
and use that $[z_{i}|z_{j}\ket = \alpha_{i}\beta_{j} - \alpha_{j} \beta_{i}$ together with the antisymmetry of $\tau_{ij}$, this integral reads 
\be
\int \prod_{i}\rd\mu(\alpha_{i})\rd\mu(\beta_{i}) e^{\sum_{i,j}( \tau_{ij} \alpha_{i}\beta_{j} +\bar{\tau}_{ij}\bar{\alpha}_{i}\bar{\beta}_{j})}
\ee
with $\rd \mu(\alpha)=e^{-| \alpha |^{2}} \rd \alpha/\pi $.
We can easily integrate over $\beta_{j}$, since the integrand is linear in $\beta_{j}$ and we obtain:
\be
\int \prod_{i}\rd\mu(\alpha_{i}) e^{\sum_{i,j,k}\alpha_{i} \tau_{ij} \bar{\tau}_{kj} \bar{\alpha}_{k}}
= \frac{1}{\det(1 + T\overline{T})}
\ee
where $T = (\tau_{ij})$ and $\overline{T} = (\overline{\tau}_{ij})$.  In the case where $n=3$ this determinant can be explicitly evaluated and it is given by
\be
\det(1 + T\overline{T}) = \left(1-\sum_{i<j} |\tau_{ij}|^{2} \right)^{2}
\ee
In the case $n=4$  the explicit evaluation  gives 
\be\label{det4}
\det(1 + T\overline{T}) = \left(1-\sum_{i<j} |\tau_{ij}|^{2} +  |R|^{2}\right)^{2}
\ee
where 
\be
R(\tau)= \tau_{12}\tau_{34} +  \tau_{13} \tau_{42}+ \tau_{14}\tau_{23}.
\ee
Note that the Pl\"ucker identity tells us that $R=0$ when $\tau_{ij} =[z_{i}|z_{j}\ket$.  

By expanding the LHS of (\ref{CC2}) for $n=4$
\bea
 \left\bra {\cal C}_{\tau_{ij}} | {\cal C}_{\tau_{ij}} \right\ket&= &\sum_{[k],[k']} \prod_{i<j} \tau_{ij}^{k_{ij}} \bar{\tau}_{ij}^{k_{ij}'} \left\bra C_{[k']} \right|\left. C_{[k]}\right\ket
 \eea
 we see that the generating functional contains information about the scalar products of the new intertwiners.
 The property of this scalar product is studied in \cite{Ljeff}.

For general $n$ we notice that
\be
  \det(1 + T\overline{T}) = \det \bpm T & 1 \\ -1 & \overline{T} \epm 
\ee
and since $T$ is $n\times n$ antisymmetric we can express the determinant as the square of a pfaffian as
\be
  \det(1 + T\overline{T}) = \left( 1 + \sum_{I} (-1)^{\frac{|I|}{2}} \mathrm{pf}(T_I)\mathrm{pf}(\overline{T_I}) \right)^2
\ee
where $I \subset \{1,...,n\}$, $|I| = 2,4,...$ up to $n$, and $T_I$ is the submatrix of $T$ consisting of the rows and columns indexed by $I$.   In particular we have  $\mathrm{pf}(T_{\{i,j\}}) = \tau_{ij}$ and  for $I = \{i,j,k,l\}$ 
\be
  R_{ijkl} \equiv \mathrm{pf}(T_{\{i,j,k,l\}}) = \tau_{ij}\tau_{kl} + \tau_{ik}\tau_{lj} + \tau_{il}\tau_{jk}.
\ee
By the pfaffian expansion formula for $|I| > 4$ $\mathrm{pf}(T_{I})$ consists of terms, all of which contain a factor $R_{ijkl}$ for some $1\leq i<j<k<l \leq n$.  For instance $\mathrm{pf}(T_{\{1,2,3,4,5,6\}}) = \tau_{12}R_{3456} - \tau_{13}R_{2456}+\cdots$.  Therefore if $\tau_{ij} =[z_{i}|z_{j}\ket$ then we have $\binom{n}{4}$ relations $R_{ijkl} = 0$ in which case the scalar product has the form
\be
 \left\bra {\cal C}_{[z_{i}|z_{j}\ket}| {\cal C}_{[z_{i}|z_{j}\ket} \right\ket 
 = \left(1-\sum_{i<j} [z_{i}|z_{j}\ket \bra z_{i}|z_{j} ] \right)^{-2} 
 = {\cal A}_{\Theta_{n}}(z_{i},\check{z}_{i})
\ee
where $ |\check{z}_{i}\ket \equiv |z_{i}]$.
This shows that  when $\tau_{ij}=[z_{i}|z_{j}\ket$, we recover the amplitude ${\cal A}$ we computed initially.
This is not a coincidence, this is always true for any graph as we now show.

\subsection{General evaluation}
\begin{definition}
Given an oriented graph $\Gamma$ we define  a generating functional  that depends holomorphically 
on  parameters $\tau_{ee'}^{v}=-\tau_{e'e}^{v}$  associated with a pair of edges $e,e'$ meeting at $v$.
\be\label{defG}
\G(\tau_{ee'}^{v}) \equiv  \int \prod_{e\in E_{\Gamma}} \rd\mu(w_{e}) \prod_{v\in V_{\Gamma}}{\cal C}^{(v)}_{\tau_{ee'}^{v}}(w_{e})
\ee
where the integral is over one spinor per edge of $\Gamma$ and we integrate a product of intertwiners for each vertex $v$.
If $v$ is a $n$-valent vertex with outgoing edges $e_{1},\cdots, e_{k}$ and  incoming edges $e_{k+1},\cdots, e_{n}$ we define
\be
{\cal C}^{(v)}_{\tau_{ee'}^{v}}(w_{e}) \equiv{\cal C}_{\tau_{ee'}^{v}}( w_{e_{1}},\cdots, w_{e_{k}},\check{w}_{e_{k+1}},\cdots, \check{w}_{e_{n}}).
\ee
\end{definition}

We then have the following lemma
\begin{lemma}
\be
\G(\tau_{ee'}^{v}) = {\cal A}_{\Gamma}(z_{e}), \quad \mathrm{if} \quad \tau_{ee'}^{v} = [z_{e}|z_{e'}\ket \quad \mathrm{when}\quad s(e)=s(e')=v
\ee
\end{lemma}
\begin{proof}
The proof is straightforward;
we start from the definition (\ref{defC}) of ${\cal C}_{\tau}$ and notice that when $ \tau_{ee'}^{v} = [z_{e}|z_{e'}\ket$ this expression reads
\bea
{\cal C}_{[z_{e}|z_{e'}\ket} (w_{e}) =\sum_{[k]} (J+1)! \widehat{C}_{[k]}(z_{e}) \widehat{C}_{[k]}(w_{e}) 
= \sum_{j_{e}} \frac{(J+1)!}{(2j_{e})!} \int \rd g [z_{e}|g|w_{e}\ket^{2j_{e}}
\eea
where we have used (\ref{CC}) in the second equality.
Integrating out $w_{e}$ and using that $$\int \rd\mu(w)[z |g_{s}|{w}\ket^{2j} [z'|g_{t}|\check{w} \ket^{2j'} =
\int \rd\mu(w)[z |g_{s}|{w}\ket^{2j} \bra w |g_{t}^{-1}|z' \ket^{2j'}=  (2j)! \delta_{j,j'}[z|g_{s}g_{t}^{-1}|z'\ket^{2j},$$ we easily obtain that
\be
\G([z_{e}|z_{e'}\ket ) = \sum_{j_{e}} \frac{\prod_{v} (J_{v}+1)!}{\prod_{e}(2j_{e})!} 
\int \prod_{v\in V_{\Gamma}} \rd g_{v} [z_{e}|g_{s_{e}}g^{-1}_{t_{e}}|z_{e^{-1}}\ket^{2j_{e}} = {\cal A}_{\Gamma}(z_{e}).
\ee
\end{proof}

We now formulate our last main result
\begin{lemma} \label{thm_gen_gauss}
The generating functional $\G$ can be evaluated as an inverse determinant 
\be
\G(\tau_{ee'}^{v}) =\frac1{\det(E-T^{\Gamma})}
\ee
where 
 \be\label{defE}
E \equiv \bpm 0 & 1 \\ -1 & 0  \epm
\ee  and $T^{\Gamma}$ is a matrix whose entries  are  labeled by half edges (or oriented edges) of $\Gamma$.  The matrix elements of $T^{\Gamma}$ are given by:
\be \label{eqn_T_tau}
 T^{\Gamma}_{e_{1} e_{2}} = \tau_{e_{1}e_{2}}^{v}\quad \mathrm{if} \quad s(e_{1})=s(e_{2})=v,
\ee
while all the other matrix elements vanish.
This matrix is skew-symmetric
\be \label{eqn_T_antisym}
T^{\Gamma}_{e_{1}e_{2}} =-T^{\Gamma}_{e_{2}e_{1}}
\ee
\end{lemma}
\begin{proof}
Two edges $e$ and $e'$ of $\Gamma$ can either share zero one or two vertices.
when two edges share a vertex there are four possible orientation of the edges at this vertex,
since each edge can be either incoming or outgoing.
 Taking all of these possibilities  
 we introduce the coefficients $\T_{ee'}$ which vanish if $s(e)$ is different from $s(e')$  is given by
 $\T_{ee'} \equiv \tau_{e e'}^{s_{e}}$ otherwise.
 If two edges meet at one vertex, one of the four coefficients 
 $\T_{ee'},\T_{e^{-1}e'^{-1}},\T_{ee'^{-1}},\T_{e^{-1}e'}$ do not vanish.
 If two edges meet at two vertices two such coefficients do not vanish.

 This matrix can be used to express explicitly the amplitude $\G$ taking into account the orientation of the edges,
 and since $[\check{w}|w'\ket= -\bra w|w'\ket$ and $[w|\check{w}'\ket= [w|w']$ and 
 $[\check{w}|\check{w}'\ket= -\bra w|w']$  the definition 
 (\ref{defG}) translates into
  \begin{align}
\G(\tau_{ee'}^{v}) 
  &= \int \prod_{e\in E_{\Gamma}} \rd\mu(w_{e}) \exp \Big\{ -\frac12 \sum_{e,e'}\left(
  \T_{e^{-1}e'} \bra w_{e}|w_{e'}\ket + \T_{e^{-1}e'^{-1}} \bra w_{e}|w_{e'}]  
  -  \T_{ee'}[ w_{e} | w_{e'} \ket  - \T_{e e'^{-1}} [ w_{e} | w_{e'} ] 
 \right)\Big\}\nonumber 
\end{align}
Note that the anti-symmetry properties  $\T_{e e'}$  is compatible with the symmetry properties of the spinor products.  
Expressing this in terms of the two components $\alpha_e,\beta_e \in \C$ of the spinor $w_e$, we get 
\begin{align}
 \G (\tau_{ee'}^{v}) 
  &= \int \prod_{e\in E_{\Gamma}} \rd\mu(\alpha_{e}) \rd\mu(\beta_{e}) 
  \exp\Big\{ -\frac12 \sum_{e,e'} 
  \Big( \T_{e^{-1}e' }(\overline{\alpha}_{e} \alpha_{e'} + \overline{\beta}_{e} \beta_{e'})  + \T_{e^{-1}e'^{-1}} ( \overline{\beta}_{e} \overline{\alpha}_{e'} -\overline{\alpha}_{e} \overline{\beta}_{e'})
   \nonumber \\
   &   \qquad \qquad \qquad \qquad\qquad \qquad \qquad \quad-  \T_{ee'}(\alpha_{e} \beta_{e'} - \beta_{e} \alpha_{e'}) - \T_{e e'^{-1}}(\alpha_{e} \overline{\alpha}_{e'} + \beta_{e} \overline{\beta}_{e'}) \Big) \Big\} 
  \nonumber \\ \nonumber
  &= \int \prod_{e\in E_{\Gamma}} \rd\mu(\alpha_{e}) \rd\mu(\beta_{e}) \exp \Big\{ -\sum_{e,e'} \Big( \overline{\alpha}_{e} A_{e e'} \alpha_{e'} + \beta_{e} B_{e e'} \alpha_{e'} + \overline{\alpha}_{e} C_{e e'} \overline{\beta}_{e'} +    {\beta}_{e} D_{e e'} \overline{\beta}_{e'} \Big) \Big\} 
\end{align}
where $\rd \mu(\alpha) = e^{-|\alpha|^2} \rd \alpha/\pi$ and
\begin{align}\nonumber
  A_{e e'} &=\frac12( \T_{e^{-1} e'}-\T_{e' e^{-1}})=  \T_{e^{-1} e'}, \qquad
  D_{e e'} = \frac12(\T_{e'^{-1} e}-\T_{e e'^{-1}})=\T_{e'^{-1} e} = A_{ee'}^{t} \\
  B_{e e'} &=  \frac12(\T_{e' e}-\T_{e e'} )= - \T_{e e'} , \qquad \nonumber\qquad \,\,\,\,\,
  C_{e e'} = \frac12(\T_{e^{-1} e'^{-1}} - \T_{e'^{-1} e^{-1}})=\T_{e^{-1} e'^{-1}}    
  \end{align}
  where $A^{t}$ denotes the transpose of $A$.
Performing the Gaussian integrations first of $\alpha$ and then of $\beta$ we get
\begin{align}
 \G(\tau_{ee'}^{v})
  &= \frac{1}{\mathrm{det}(1+ A)} \int \prod_{e\in E_{\Gamma}} \rd\mu(\beta_{e}) \exp\Big\{ - \sum_{e,e'} \Big({\beta}_{e} A^{t}_{e e'}  \overline{\beta}_{e'} - \beta_{e} (B(1+A)^{-1}C)_{ee'}
  \overline{\beta}_{e'} \Big) \Big\} \\
  &= \mathrm{det}(1+A)^{-1} \mathrm{det}\left( 1 + A^{t} - B(1+A)^{-1} C \right)^{-1} \\
  &= \mathrm{det}\bpm 1+A & 0 \\ B & 1  \epm^{-1} \mathrm{det}\bpm 1 & (1+A)^{-1}C \\ 0 & 1 + A^{t} - B(1+A)^{-1}C \epm^{-1}\\
  &= \mathrm{det}\bpm 1+A & C \\ B & 1 + A^{t} \epm^{-1} 
\end{align}
The  matrix $E$ introduced in (\ref{defE}) 
has a unit determinant; thus the previous determinant is also equal to the determinant of the antisymmetric matrix
\be
\mathrm{det}\left[E \bpm 1+A & B \\ C & 1 + A^{t} \epm  \right]^{-1} =
\mathrm{det}\bpm B  & (1+A^{t}) \\ -(1+A) & -C \epm^{-1} 
=\det(E - \T)
\ee
which is what we desire to establish.
\end{proof}
We now are going to evaluate explicitly this determinant in much the same way as Theorem \ref{thm_amp}.  In order to do so we must define the following quantities.
\begin{definition}
A simple loop of $\Gamma$ is a loop of $\Gamma$ in which each edge enters at most once.  We say two
simple loops are disjoint if they have no edges in common.
Given a simple loop $\ell= \{e_{1}, \cdots ,e_{n}\}$ we define the quantity
\be \label{eqn_A_loop}
A_{\ell}(\tau) = -(-1)^{|e|} \tau_{e_{1}^{-1} e_2}^{s(e_2)} \tau_{e_{2}^{-1} e_3}^{s(e_3)} \cdots \tau_{e_{n}^{-1} e_1}^{s(e_1)}
\ee
where $|e|$ is the number of edges of $l$ whose orientation agrees with the chosen orientation of  $\Gamma$.  
Finally, given a collection of disjoint simple loops $L= l_{1},..., l_{k}$ we define
\be 
A_{L}(\tau) = A_{\ell_{1}}(\tau)\cdots A_{\ell_{k}}(\tau).
\ee
\end{definition}
With these definitions the generating functional is given by
\begin{theorem} \label{thm_gen}
\be
\G(\tau) = \frac1{\left(1 + \sum_{L} A_{L}(\tau)\right)^{2}}
\ee
where the sum is over all collections of disjoint simple loops of $\Gamma$.
\end{theorem}
Note that this result for the generating functional $\G(\tau)$ is very similar to the first theorem \ref{thm_amp} we established in the first section
for the coherent amplitude ${\cal A}_{\Gamma}(z_{e})$  .
The key difference is that the general amplitude involve a sum over simple loops which contains cycles or non intersecting simple loops, but also 
simple loops that intersect at a vertex. The relation between the two theorems comes from the fact that if the Pl\"ucker  relation is satisfied then the sum of 
loops that meet at this vertex vanish. This can be easily seen graphically in Fig. \ref{fig_STU}
and it is established algebraically in the appendix.  This allows us to  offer an alternative proof of Theorem \ref{thm_amp}
as a corollary to Theorem \ref{thm_gen}.
\begin{corollary} \label{cor_gen}
If $\tau_{ee'}^{v} = [z_{e}|z_{e'}\ket$ where $s(e)=s(e')=v$ then
\be
{\cal G}_{\Gamma}(\tau) = \frac1{\left(1 + \sum_{C} A_{C}(\tau)\right)^{2}}, 
\ee
where the sum is over all disjoint cycle unions of $\Gamma$.
\end{corollary}
Again the proof of this corollary can be found in the appendix.  It is interesting to note the similarity between the proof of Lemma \ref{thm_sc_det_row_operation} and the proof of Corollary \ref{cor_gen}.

\begin{figure} 
  \centering
    \includegraphics[width=1\textwidth]{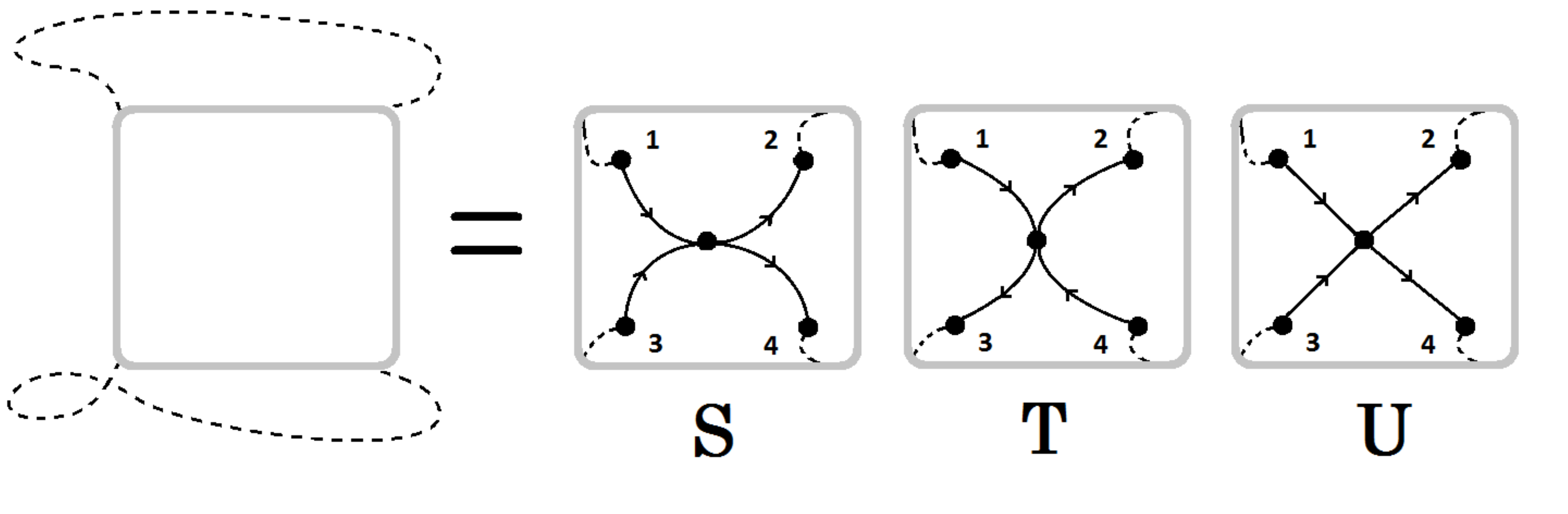}
    \caption{A simple loop depicted by the dashed line intersects itself at a vertex within the box.  In fact there three possible collections of inequivalent simple loops which intersect at this vertex and have the same unoriented edges in common.  These three collections correspond to the three orientations S,T, and U of the four edges meeting at this vertex.  Note the following identification of vertices: S=(12)(34), T=(13)(42), U=(14)(32) which is an allusion to the Pl\"ucker relation.  An algebraic proof of how the amplitudes of intersecting simple loops arrange into the Pl\"ucker form is given in the proof of Corollary \ref{cor_gen}. }  \label{fig_STU}
\end{figure}

\acknowledgments

We would like to thank Etera Livine and Valentin Bonzom for helpful discussions.  Research at Perimeter Institute is supported by the Government of Canada through Industry Canada
and by the Province of Ontario through the Ministry of Research and Innovation.  JH would like to thank the Natural Sciences and Engineering Research Council of Canada
(NSERC) for his post graduate scholarship.

\appendix

\section{Invariant integration of a homogeneous function}

Given a spinor $|\alpha \ket$ we define the U$(2)$ group element $g(\alpha) = |0\ket \bra \alpha| + |0][\alpha|$ where $g(\alpha)g(\alpha)^\dagger = \bra \alpha | \alpha \ket$.  Suppose that  $F(g(\alpha))$ is a homogeneous function of $|\alpha \ket$ of degree $2J$, that is 
$F(g(\lambda \alpha)) = \lambda^{2J}F( g(\alpha))$. Then in the pseudo-spherical coordinates
\be 
  | \alpha \ket = \bpm  r \cos(\theta) e^{i \phi} \\  r \sin(\theta) e^{i \psi} \epm
\ee
where $r \in (0, \infty)$, $\theta \in [0,\pi/2)$, $\phi \in [0,2\pi)$, $\psi \in [0,2\pi)$ we have
\be
  F(g(\alpha)) = r^{2J} F(g(\widetilde{\alpha}))
\ee
where
\be 
  |\widetilde{\alpha} \ket = \bpm \cos(\theta) e^{i \phi} \\  \sin(\theta) e^{i \psi} \epm
\ee
and $g(\widetilde{\alpha}) \in \SU(2)$.  The Lebesgue measure in these coordinates is $\rd^{4}|\alpha\ket = r^3 \sin(\theta) \cos(\theta) dr \wedge d\phi \wedge d\theta \wedge d\psi$ and so
\be
 \int_{\C^{2}} \rd^{4}|\alpha\ket   e^{-\bra \alpha | \alpha \ket} F(g(\alpha))
= \int_{0}^{\infty} \rd r r^{3+2J} e^{-r^2} \int_{0}^{\pi/2} \rd \theta \sin(\theta) \cos(\theta) \int_{0}^{2\pi} \rd \phi \int_{0}^{2\pi} \rd \psi F(g(\widetilde{\alpha})).
\ee
Performing the integration over $r$
\be
  \int \rd r r^{3+2J} e^{-r^2} = \frac{1}{2} \Gamma(J+2)
\ee
gives
\be
 \int_{\C^{2}} \rd^{4}|\alpha\ket   e^{-\bra \alpha | \alpha \ket} F(g(\alpha))
= \Gamma(J+2) \int_{SU(2)}\rd g  F(g)
\ee
where $\rd g$ is the normalized Haar measure on SU$(2)$.

\section{Proof of Theorem \ref{thm_amp}}

Recall the Laplace expansion of the determinant for a $n \times n$ matrix (of complex numbers)
\be
  \mathrm{det}(A) = \sum_{\pi} \text{sgn}(\pi) a_{1 \pi(1)} a_{2 \pi(2)} \cdots a_{n \pi(n)}.
  \label{eqn_Laplace_det}
\ee
An equivalent definition of the determinant can be given in terms of cycle covers of a complete directed graph on $n$ vertices \cite{Mahajan}.  On a complete graph we can label a loop by a sequence of vertices since there is only one edge between any two vertices.  A cycle is defined to be a simple loop for which all the vertices are distinct and a cycle cover is a collection of cycles which covers all the vertices in the graph, i.e. all of $\{1,...,n\}$.  Notice that every permutation of $(1,...,n)$ corresponds to a unique partition of the set $\{1,...,n\}$ into disjoint cycles.  For example the permutation 
\be
  \pi =
  \bpm 
    1 & 2 & 3 & 4 & 5 & 6 \\
    2 & 4 & 6 & 1 & 5 & 3
  \epm
\ee
corresponds to the cycle cover $\cC = (124)(36)(5)$.  The weight of a cycle $C=(c_{1}\cdots c_{i})$ is defined to be $W(C) = a_{c_1 c_2} a_{c_2 c_3} ... a_{c_i c_1}$ and the weight of a cycle cover is the product of the weights of its cycles
\footnote{The weight of a loop is defined in the same way.}.
Furthermore, it can be shown that the sign of a permutation is equal to $(-1)^{n+k}$ where $k$ is the number of cycles in its corresponding cover.   Therefore, Eq. (\ref{eqn_Laplace_det}) can be written as
\be
  \mathrm{det}(A) = \sum_{\cC} \mathrm{sgn}(\cC) W(\cC).
  \label{eqn_det_cycles}
\ee
Now suppose that the matrix $A$ is composed of elements which are noncommutative such as $2 \times 2$ matrices in the case of Eq. (\ref{eqn_matrix_X}).  In this case we lose many useful relations of the determinant such as the multiplicative property and the behavior with respect to elementary row operations due to the noncommutativity.  Yet for special types of matrices which we call scalar loop matrices we can define a quasi-determinant for which these properties still hold.

\begin{definition}
A matrix is called a scalar loop matrix if for any loop $L$ the quantity $S(L) = \frac{1}{2}(W(L) + W(L^{-1}))$ is scalar where $L$ and its inverse $L^{-1}$ begin with the same element but the sum is otherwise invariant under cyclic permutations of $L$. 
\end{definition}
\begin{definition}
Let $A$ be a $n$ by $n$ scalar loop matrix.  The loop determinant of $A$ is defined to be
\be
  \mathrm{Ldet}(A) = \sum_{\cC} \mathrm{sgn}(\cC) S(\cC)
  \label{eqn_sc_det_cycles}
\ee
where the sum is over all cycle covers $\cC = C_1...C_k$ on $\{1,..,n\}$.
\end{definition}
Note that for a commutative matrix $A$ Eq. (\ref{eqn_sc_det_cycles}) is equivalent to Eq. (\ref{eqn_det_cycles}) for which the multiplicative property was studied in \cite{Kovacs}.  The reason we are interested in this definition is because of the following observation.
\begin{lemma}
The matrix $1+X$ in Eq. (\ref{eqn_matrix_X}) is a scalar loop matrix.
\end{lemma}
\begin{proof}
First suppose $\Gamma$ is a complete oriented graph so that we can continue to label loops by pairs of vertices and let $L = (l_1 l_2 \cdots l_i)$ be a loop on $\{1,...,n\}$.  Then $X_{l_j l_k} = |z_{l_j l_k} \ket [z_{l_k l_j} |$ if the edge from $l_j$ to $l_k$ is positively oriented and the negative otherwise.  Suppose that $L$ has $|e|$ edges which are opposite the orientation.  Then
\be
  W(L) = (-1)^{|e|} |z_{l_1 l_2} \ket [z_{l_2 l_1} | z_{l_2 l_3} \ket \cdots [z_{l_{i} l_{i-1}} | z_{l_i l_1} \ket [ z_{l_1 l_i} |
\ee
and 
\be
  W(L^{-1}) = (-1)^{|e|+i} |z_{l_1 l_{i}} \ket [z_{l_{i} l_1} | z_{l_{i} l_{i-1}} \ket \cdots [z_{l_{2} l_{3}} | z_{l_2 l_{1}} \ket [ z_{l_1 l_2} |
\ee
Now using the identity $[z | w \ket = - [w | z \ket$ we have an extra factor of $(-1)^{i-1}$ in the second term and so
\be
  W(L) + W(L^{-1}) = (-1)^{|e|} [z_{l_2 l_1} | z_{l_2 l_3} \ket \cdots [z_{l_{i} l_{i-1}} | z_{l_{i} l_1} \ket \Big( |z_{l_1 l_2} \ket [ z_{l_1 l_i} | - |z_{l_1 l_i} \ket [ z_{l_1 l_2} | \Big)
\ee
now using $|z \ket [w| - |w \ket [z| = -[z|w\ket \one$ we have
\be
  S(L) \equiv \frac{1}{2}\left(W(L) + W(L^{-1})\right) = \frac{(-1)^{|e|}}{2} [z_{l_1 l_i} | z_{l_1 l_2} \ket [z_{l_2 l_1} | z_{l_2 l_3} \ket \cdots [z_{l_{i} l_{i-1}} | z_{l_{i} l_{1}} \ket \one
\label{eqn_weight_cycle}
\ee
By writing $X_{ij}$ as in Eq. (\ref{eqn_matrix_X}) we generalize $\Gamma$ to have any number of edges between pairs of vertices. In that case it is clear that $S(L)$ is equal to the sum of weights of the form on the r.h.s. of Eq. (\ref{eqn_weight_cycle}) over all loops in $\Gamma$ traversing the vertices $(l_1 l_2 \cdots l_i)$ in order.
\end{proof}
The purpose of this section is to prove the following lemma.
\begin{lemma}\label{Ldet=det}
Let $A$ be a scalar loop matrix composed of block matrices and denote the ordinary determinant by $|A|$.  Then
\be
  |A| = \left|\mathrm{Ldet}(A)\right|
\ee
\end{lemma}
This lemma follows from the following property of the loop determinant.
\begin{proposition}
Let $A$ be a scalar loop matrix.  Then the loop determinant behaves as the usual determinant under all the elementary row operations.
In particular the addition of a scalar multiple of one row of $A$ to another row leaves the loop determinant invariant.
\label{thm_sc_det_row_operation}
\end{proposition}
\begin{proof}
Suppose we add a scalar multiple $\lambda$ of row $i$ of $A$ to row $j$.  Then Eq. (\ref{eqn_sc_det_cycles}) is changed by replacing the single factor $A_{i \cdot}$ in each weight by $A_{i \cdot} + \lambda A_{j \cdot}$.  Therefore Eq. (\ref{eqn_sc_det_cycles}) becomes a sum of its original terms plus terms proportional to $\lambda$.  We will now show that all terms proportional to $\lambda$ cancel each other.

Let $\cC$ be a cycle cover of $1,...,n$.  Then there exists two possibilities: $i$ and $j$ are in the same cycle or $i$ and $j$ are in different cycles. Suppose that they are in the same cycle $C$ and let $\cC^\prime$ be the rest of $\cC$.  By cyclic invariance we can assume that $i = c_1$ and call $j=c_j$ where $C = (c_1 ... c_j... c_N)$. Replacing $A_{c_1 c_2}$ with $A_{c_{1} c_{2}} + \lambda A_{c_{j} c_{2}}$ in $W(C)$ we get 
\begin{align}
  W(C) 
  \rightarrow (A_{c_1 c_2} + \lambda A_{c_{j} c_{2}}) A_{c_2 c_3} \cdots A_{c_{j-1}c_{j}} A_{c_{j} c_{j+1}} \cdots A_{c_{N}c_{1}} 
  = W(C) + \lambda W(\widetilde{C}) N(C)
\end{align}
where $\widetilde{C} = (c_j c_{2} c_{3} ... c_{j-1})$ and $N(C) = A_{c_{j} c_{j+1}} A_{c_{j+1} c_{j+2}} \cdots A_{c_{N}c_{1}}$.  Now consider the cycle $\widehat{C} = (c_1 c_{j+1} c_{j+2} ... c_N)$ then
\begin{align}
  W(\widehat{C}) 
  \rightarrow (A_{c_1 c_{j+1}} + \lambda A_{c_{j} c_{j+1}}) A_{c_{j+1} c_{j+2}} \cdots A_{c_{N}c_{1}} 
  = W(\widehat{C}) + \lambda N(C)
\end{align}
and moreover
\be
 W(\widetilde{C}) W(\widehat{C})
  \rightarrow W(\widetilde{C})W(\widehat{C}) + \lambda W(\widetilde{C}) N(C)
\ee
This demonstrates that $W(C)$ and $W(\widetilde{C})W(\widehat{C})$ produce terms proportional to $\lambda$ which are equal but have opposite sign in Eq. (\ref{eqn_sc_det_cycles}) since $\text{sgn}(\widetilde{C} \widehat{C}) = -\text{sgn}(C)$.  We now show exactly how these terms cancel in Eq. (\ref{eqn_sc_det_cycles}), by considering eight cycle covers for which the terms proportional to $\lambda$ all cancel eachother.  Indeed, let $C_1 = (c_{1} c_{2}  ... c_{j-1})$, $C_2 = (c_{j} c_{j+1} c_{j+2} ... c_{N})$, $C_3 = (c_{1} c_{N} c_{N-1} ... c_{j+1})$, $C_4 = (c_{j} c_{j-1} c_{j-2}... c_{2})$, $C_5 = (c_{1} c_{2} ... c_{j-1} c_{j} c_{j+1} ... c_{N})$, $C_6 = (c_{1} c_{2} ... c_{j-1} c_{j} c_{N} c_{N-1} ... c_{j+1})$, $C_7 = (c_{1} c_{j-1} c_{j-2} ... c_{2} c_{j} c_{j+1} ... c_{N})$, $C_8 = (c_{1} c_{j-1} c_{j-2} ... c_{2} c_{j} c_{N} c_{N-1} ... c_{j+1})$ then it is straightforward to show that
\be
  S(C_1)S(C_2) + S(C_3)S(C_4) - S(C_5) - S(C_6) - S(C_7) - S(C_8)
\ee
is invariant after the row operation, i.e. the terms proportional to $\lambda$ cancel.  Conversely, if $c_1$ and $c_j$ are in different cycles we can write them as $C_1$ and $C_2$ in which case we can construct $C_3$,..., $C_8$ which leads to the same cancellation.

It is easy to see from Eq. (\ref{eqn_det_cycles}) that multiplying a row by a scalar produces an overall factor of $\lambda$ and switching two rows produces a minus sign, just like the determinant over a field.  Hence the loop determinant behaves as one would expect under all the elementary row operations.
\end{proof}
We can now give the proof of lemma \ref{Ldet=det} by induction.
\begin{proof}
By Theorem \ref{thm_sc_det_row_operation} the loop determinant is unchanged after Gaussian elimination so after eliminating the first column
\be
  \mathrm{Ldet}(A) 
  = \mathrm{Ldet}
  \bpm 
  A_{11} & A_{12} & \ldots & A_{1n} \\
  A_{21} & A_{22} & \ldots & A_{2n} \\
  \vdots & \vdots & \ddots & \vdots \\
  A_{n1} & A_{n2} & \ldots & A_{nn}
  \epm 
  = \mathrm{Ldet}
  \bpm
  A_{11} & A_{12} & \ldots & A_{1n} \\
  0 &    &   &    &  \\
  \vdots &   & B  &  \\
  0 &    &   &    &
  \epm
\ee
where $B$ is a $(n-1) \times (n-1)$ matrix with entries $B_{ij} = A_{ij} - A_{i1} A_{11}^{-1} A_{1j}$.  Note that since $A$ is a scalar loop matrix $A_{11}$ is scalar so $A_{11}^{-1}$ does indeed exist and is also scalar.  Furthermore, if $L=(l_1 l_2 ... l_i)$ is a loop of $\{2,3,...,n\}$ then $W_B(L) = B_{l_1 l_2} B_{l_2 l_3} \cdots B_{l_i l_1}$ can be expressed as
\be
  W_B(L) = W_A(L) + \sum_{\sigma} (-1)^{|\sigma|} W_{A}(L(\sigma))
\ee
where $\sigma \subset \{1,2,...,i\}$ and $L(\sigma) = (l_1 ... l_{\sigma_1} 1 l_{\sigma_1 + 1} ... l_{\sigma_2} 1 l_{\sigma_2+1} ... l_{i})$, i.e. it is $L$ with 1 inserted after every element of $\sigma$.  In other words $L(\sigma)$ is a loop of $\{1,2,3,...,n\}$ and so $S_B(L)$ is scalar which shows that $B$ is a scalar loop matrix.  

The hypothesis is clearly true for $n=1$ so now assume it is true for scalar loop matrices of size $(n-1) \times (n-1)$.  Then $|B| = \left|\mathrm{Ldet}(B)\right|$ which then implies
\be
  |A| = |A_{11}| \cdot |B| = |A_{11}| \cdot \left|\mathrm{Ldet}(B)\right| = \left|\mathrm{Ldet}(A)\right|
\ee 
which advances the induction hypothesis.
\end{proof}
Finally we apply the previous lemmas to the matrix $1+X$ in Eq. (\ref{eqn_matrix_X}).
\begin{lemma}
$|1+X| = \left(1 + \sum_{C} A_{C}(z_{e})\right)^{2}$ where the sum is over all disjoint cycle unions $C$ of $\Gamma$ and $A_{C}(z_e)$ is defined in Eq. (\ref{eqn_cycle_union_amp}).
\end{lemma}
\begin{proof}
By the previous lemmas 
\be
  |1+X| = \left|\mathrm{Ldet}(1+X)\right| = \left( \sum_{\cC} \mathrm{sgn}(\cC) S(\cC) \right)^2
\ee
where the sum is over all cycle covers of $V_\Gamma$.  Since the loop determinant is a scalar (proportional to the 2 by 2 identity), its determinant is a perfect square.  The 1-cycles of $1+X$ correspond to the diagonal which all have weight 1.  The cycle cover of all 1-cycles produces the term equal to unity.  The 2-cycles of $1+X$ all vanish since $[z_e|z_e \ket=0$. Therefore the cycle covers consist of disjoint unions of non-trivial cycles with the remaining vertices covered by 1-cycles.  This is enough to see that the weight from the loop determinant formula agrees with the weight in Eq. (\ref{eqn_cycle_union_amp}).  Now the sign of each term is $(-1)^{n+k}$ from the cycle cover and $(-1)^{|e|}$ from the weight formula in Eq. (\ref{eqn_weight_cycle}).  If a cycle cover has $i$ non-trivial cycles covering $n-r$ vertices then there are $k = i+r$ cycles in the cover.  Thus if we assign $(-1)^{|n|+|e|+1}$ to each non-trivial cycle where $|n|$ is the number of vertices in the cycle then $\sum (|n|+1) = (n-r) + i  = n+k-2r$ which agrees with the weight from the cycle cover.
\end{proof}
Now Theorem \ref{thm_amp} follows trivially from the last lemma.

\section{Proof of Theorem \ref{thm_gen}}
We now want to evaluate the determinant of $E - \T$. This is a anti-symmetric matrix of size $2N$ by $2N$ 
indexed by $e_1,...,e_N, e_{1}^{-1},...,e_{N}^{-1}$. Therefore this determinant can be evaluated as the 
square of the pfaffian of $E - \T$.
We cannot directly evaluate the Pfaffian of a matrix as a sum over cycles, however it is possible  following \cite{Gunter} 
to write the product of pfaffians of two $2N$ by $2N$ antisymmetric matrices as
\be
  \mathrm{pf} A \cdot \mathrm{pf} B = \sum_{\text{C}} (-1)^{k} W_{A,B}(C)
\ee
where the sum is over cycle covers $C = c_1, ..., c_k$ of $\{1,...,2N\}$ having $k$ cycles and where each cycle is of even length.  
The weight of a cycle cover is the product of the weights of its cycles and the weight of a single cycle $c = (i_1, ..., i_{n})$ with $i_1 > i_2,...,i_{n}$ is given by 
\be
W_{A,B}(c) =   A_{i_1 i_2} B_{i_2 i_3} A_{i_3 i_4} B_{i_4 i_5} ... A_{i_{n-1} i_{n}} B_{i_n i_1}.
\ee
The specification of $i_1$ as the largest element in the cycle avoids any ambiguity in the definition of the weight.  
If one chooses $B =E$ then $\mathrm{pf}E=(-1)^{N(N-1)/2}$ then we have an expression for $\mathrm{pf}A$ in terms of cycle covers up to an overall sign.  
Let us therefore  set $A =E-T^{\Gamma}$ and let us choose $B=E$.

Lets start by evaluating the weight of a 2-cycle.  Since $E_{ij}$ is non-vanishing only if $j = i \pm N$ the weight must have the form
\be
A_{i_{1}+N, i_{1}} E_{i_{1}, i_{1}+N} = (E- \T)_{e_{1}^{-1}e_{1}} = - (1+\T_{e_{1}^{-1}e_{1}}).
\ee
Note that $\T_{e^{-1}e}\neq 0$ only if $e$ forms a 1-cycle (or bubble) at a vertex of $\Gamma$, i.e. $s(e)=t(e)$.
We have used the correspondence between $i_{1}= e_{1}$  and $i_{1}+N=e^{-1}_{1}$ if $i_{1}<N$.
This shows that $2$-cycles of $\{1,...,2N\}$ correspond to an evaluation in terms of 1-cycles of $\Gamma$.

Lets now consider a 4-cycle of $\{1,...,2N\}$. There are two possibilities depending on whether the second index is $i_{2}$ or $i_{2}+N$.
In the first case we get
\be \label{eqn_2_cycle_1}
 A_{i_1+N,  i_2} E_{i_2, i_2 + N} A_{i_2 + N, i_1} E_{i_1, i_1+N}=  \T_{e_{1}^{-1} e_{2}} \T_{e_{2}^{-1}e_{1}}.
\ee
In the second case we have
\be \label{eqn_2_cycle_2}
 A_{i_1+N,  i_2+N} E_{i_2+N, i_2} A_{i_2,  i_1} E_{i_1, i_1+N}= -\T_{e_{1}^{-1} e_{2}^{-1}} \T_{e_{2}e_{1}}.
\ee
In both cases we have used the fact that since $c$ is a cycle we necessarily have $i_{1} \neq i_{2}$.  Hence (because of the presence of $B_{i_{1},i_{1}+N}$)
we have that $e_{1} \neq e_{2}^{-1}$. This means that we can replace the element $ (E -  \T)_{e_{2}e_{1}}$ by $- \T_{e_{2}e_{1}}$.
One can now see that these weights correspond to 2-cycles of $\Gamma$.  The first case corresponds to the cycle of edges  $(e_{1}e_{2})$ while the second case corresponds to $(e_{1}e_{2}^{-1})$.  Clearly at most one of (\ref{eqn_2_cycle_1}) and (\ref{eqn_2_cycle_2}) is nonvanishing, since at most two of the elements of $\T$ are nonvanishing depending on the orientation.  
The difference in sign comes from $B_{i_2+N, i_2}=-1$ while $B_{i_1, i_1+N}=B_{i_2, i_2+N}=1$.  In effect we obtain a minus sign for each edge that { disagrees} with the orientation of $\Gamma$,
we also get a minus sign for every edge.  

This result generalizes easily now to the case of a $2n$-cycle of $\{1,...,2N\}$.
The same reasoning shows that the weight 
\be
W_{A,B}(c) = A_{i_1+N, i_2} E_{i_2, i_2\pm N} A_{i_2\pm N, i_3} E_{i_3, i_3\pm N}\cdots A_{i_{n-1}\pm N, i_n} E_{i_1, i_1+ N}.
\ee
is non zero if and only if the sequence of edges  $(e_{1},\cdots, e_{n})$ corresponds to a simple loop $\ell$ of $\Gamma$ of length $n$.
In that case 
\be 
W_{A,B}(c) = (-1)^{n- |\bar{e}|} \T_{e_{1}^{-1} e_{2}} \T_{e_{2}^{-1} e_{3}} \cdots \T_{e_{n}^{-1} e_{1}} =  - A_{\ell}(\tau)
\ee
and $|\bar{e}|$ is the number of times $i_j > N$ in which case $B_{i_j, i_j-N} = -1$.  Again this corresponds to traversing the edge $e_{j}$ in the  orientation opposite to the one  of $\Gamma$
 thus $|\bar{e}|$ is the number of edges in $c$ which { disagrees} with the orientation of $\Gamma$.
 Not that if we denote $|e| =n-|\bar{e}|$ is the number of edges that { \it agrees} with the orientation of $\Gamma$.
This establish therefore the correspondence between $2n$-cycles $c$ of $\{1,...,2N\}$ and simple loops of $\Gamma$ of length $n$, moreover the amplitude for a simple cycle 
is precisely minus the amplitude of the loop in $\Gamma$.

A cycle cover $\cC$ on $\{1,...,2N\}$ consists of a disjoint union of 2-cycles and non-trivial (i-e the cycles which are not 2-cycles)  cycles of $\{1,...,2N\}$.
We established that each 2-cycle of $\{1,...,2N\}$ as a weight in the sum given by $(1+T_{e_{i}^{-1}e_{i}})$ where $(e_{i}^{-1}e_{i})$ correspond to a bubble in $\Gamma$.
We also established that each nontrivial cycle on $\{1,...,2N\}$ (with non-zero weight) corresponds to a simple loop of $\Gamma$ with amplitude $A_{\ell}$.
This shows that $\mathrm{pf}(E-\T)$ is (up to an overall sign) equal to
$$\sum_{L} \prod_{v\notin L} \left(\prod_{s(e)=v=t(e)}(1+ T_{e^{-1}e})\right) A_{L}(\tau) $$
where the sum is over disjoint union of simple loops of length at least 2 and the product is over all
vertices not in $L$, with a weight given by the product over the bubbles touching $v$ (and with the convention that the weight is $1$ if there is no bubbles).
Now if $T_{e^{-1}e}$ is non zero this means that $(e^{-1}e)$ is appositively oriented bubble; that is a simple loop of length 1.
Therefore expanding the previous product we get that the pfaffian of $(1+\T)$ is (up to an overall sign) equal to
\be
\sum_{L} A_{L}(\tau) \ee
where the sum is over disjoint union of simple loops of  any length.
Which is what we desired to establish.

\section{Proof of Corollary \ref{cor_gen}}

Suppose a simple loop $U = (e_1e_{2} \cdots e_{i-1} e_i\cdots e_{n-1}e_n)$ is such that $s(e_1) = s(e_i) = v$ and $t(e_{i-1}) = t(e_{n})=v$, i.e. it intersects itself at the vertex $v$.  
Then there exists another simple loop $T = (e_1e_{2} \cdots e_{i-1} e_{n}^{-1} e_{n-1}^{-1} \cdots  e_{i}^{-1})$ which also intersects itself at $v$.  Lastly, there exists a pair of simple loops $S = (e_1... e_{i-1})(e_i ... e_n)$ which share the vertex $v$.  The triple $S,T,U$ exhaust the collections of disjoint simple loops which have an intersection at $v$ and contain precisely the set of edges $\{e_1,...,e_n\}$.

Suppose that $p_1$ edges of $\{e_1,...,e_{i-1}\}$ and $p_2$ of $\{e_i,...,e_n\}$ disagree with the orientation of $\Gamma$.  
And lets introduce the amplitudes
\bea
T_{e_{1}\cdots e_{i-1}} \equiv \left(\tau_{e_{1}^{-1} e_{2}}^{s(e_{2})} \cdots \tau_{e_{i-2}^{-1} e_{i-1}}^{s(e_{i-1})}\right)
\eea
Then by the prescription (\ref{eqn_A_loop})
\bea
  A_U = (-1)^{p_1+p_2+1}\, T_{e_{1}\cdots e_{i-1}} \tau_{e_{i-1}^{-1} e_{i}}^{v}
 T_{e_{i}\cdots e_{n}}\tau_{e_{n}^{-1} e_{1}}^{v} \\
  A_T = (-1)^{p_1+p_2+n-i} \,T_{e_{1}\cdots e_{i-1}} \tau_{e_{i-1}^{-1} e_{n}^{-1}}^{v} T_{e_{n}^{-1}\cdots e_{i}^{-1}} \tau_{e_{i} e_{1}}^{v} \\
  A_S = (-1)^{p_1+p_2} \, \tau_{e_{1}\cdots e_{i-1}} \tau_{e_{i-1}^{-1} e_{1}}^{v} T_{e_{i}\cdots e_{n}} \tau_{e_{n}^{-1} e_{i}}^{v} 
\eea
Using the antisymmetry property  of $\tau$ shows that $ (-1)^{n-i}T_{e_{n}^{-1}\cdots e_{i}^{-1}} = T_{e_{i}\cdots e_{n}}$ 
Thus
\bea
  A_S + A_T + A_U = (-1)^{p_1 + p_2} \, T_{e_{1}\cdots e_{i-1}} T_{e_{i}\cdots e_{n}}
  \nonumber  \left( \tau_{e_{i-1}^{-1} e_{1}}^{v} \tau_{e_{n}^{-1} e_{i}}^{v} + \tau_{e_{i-1}^{-1} e_{n}^{-1}}^{v} \tau_{e_{i} e_{1}}^{v} - \tau_{e_{i-1}^{-1} e_{i}}^{v} \tau_{e_{n}^{-1} e_{1}}^{v}  \right)
\eea
For clarity let $1 = e_{i-1}^{-1}$, $2 = e_1$, $3 = e_{n}^{-1}$, and $4 = e_i$ then the last factor
\be
  \left( \tau_{12}^{v} \tau_{34}^{v} + \tau_{13}^{v} \tau_{42}^{v} - \tau_{14}^{v} \tau_{32}^{v}  \right)
\ee
is the Pl\"ucker relation and vanishes under the hypothesis.  Hence the only collections of simple loops which survive
the identification $\tau_{ee'}=[z_{e}|z_{e'}\ket$ are ones which are non-intersecting and do not share vertices with other simple loops, i.e. they are  disjoint unions of non-trivial cycles.



\begin{thebibliography}{99}

\bibitem{AL} 
  A.~Ashtekar and J.~Lewandowski,
  ``Background independent quantum gravity: A Status report,''
  Class.\ Quant.\ Grav.\  {\bf 21}, R53 (2004)
  [gr-qc/0404018].
  \bibitem{PR}
  G. Ponzano; T. Regge, 
  ``Semiclassical limit of Racah coefficients'',
   p1-58, in: Spectroscopic and group theoretical methods in physics, ed. F. Bloch, North-Holland Publ. Co., Amsterdam, 1968.
   
   \bibitem{EPRL}
J. Engle, E. Livine, R. Pereira and C. Rovelli,
{\it LQG vertex with finite Immirzi parameter},
Nucl.Phys.B799 (2008) 136-149 [arXiv:0711.0146]

\bibitem{FK}
L. Freidel and K. Krasnov,
{\it  A New Spin Foam Model for 4d Gravity},
Class.Quant.Grav.25 (2008) 125018 [arXiv:0708.1595]

 \bibitem{BarrettAs1} 
  J.~W.~Barrett, R.~J.~Dowdall, W.~J.~Fairbairn, H.~Gomes and F.~Hellmann,
  ``Asymptotic analysis of the EPRL four-simplex amplitude,''
  J.\ Math.\ Phys.\  {\bf 50}, 112504 (2009)
  [arXiv:0902.1170 [gr-qc]].

\bibitem{BarrettAs2} 
  J.~W.~Barrett, W.~J.~Fairbairn and F.~Hellmann,
  ``Quantum gravity asymptotics from the SU(2) 15j symbol,''
  Int.\ J.\ Mod.\ Phys.\ A {\bf 25}, 2897 (2010)
  [arXiv:0912.4907 [gr-qc]].

\bibitem{FC1} 
  F.~Conrady and L.~Freidel,
  ``On the semiclassical limit of 4d spin foam models,''
  Phys.\ Rev.\ D {\bf 78}, 104023 (2008)
  [arXiv:0809.2280 [gr-qc]].
  
\bibitem{FC2} 
  F.~Conrady and L.~Freidel,
  ``Quantum geometry from phase space reduction,''
  J.\ Math.\ Phys.\  {\bf 50}, 123510 (2009)
  [arXiv:0902.0351 [gr-qc]].
  
  \bibitem{poly} 
  E.~Bianchi, P.~Dona and S.~Speziale,
  ``Polyhedra in loop quantum gravity,''
  Phys.\ Rev.\ D {\bf 83}, 044035 (2011)
  [arXiv:1009.3402 [gr-qc]].
  \bibitem{twistedgeo}
L. Freidel and S. Speziale,
{\it Twisted geometries: A geometric parametrisation of SU(2) phase space},
arXiv:1001.2748
  
\bibitem{UN1} 
  L.~Freidel and E.~R.~Livine,
  ``The Fine Structure of SU(2) Intertwiners from U(N) Representations,''
  J.\ Math.\ Phys.\  {\bf 51}, 082502 (2010)
  [arXiv:0911.3553 [gr-qc]].
 
 \bibitem{UN2} 
  L.~Freidel and E.~R.~Livine,
  ``U(N) Coherent States for Loop Quantum Gravity,''
  J.\ Math.\ Phys.\  {\bf 52}, 052502 (2011)
  [arXiv:1005.2090 [gr-qc]].

  \bibitem{UN3} 
  M.~Dupuis and E.~R.~Livine,
  ``Holomorphic Simplicity Constraints for 4d Spinfoam Models,''
  Class.\ Quant.\ Grav.\  {\bf 28}, 215022 (2011)
  [arXiv:1104.3683 [gr-qc]].
  
  \bibitem{LD} 
  M.~Dupuis and E.~R.~Livine,
  ``Holomorphic Simplicity Constraints for 4d Riemannian Spinfoam Models,''
  arXiv:1111.1125 [gr-qc].
  
  \bibitem{Penrose}
  R. Penrose, ``Applications of negative dimensional tensors'',
  Advances in twistor theory, Huston and Ward eds. Research notes in mathematics, Pitman Publ. 308-312.\\
  R. Penrose, ``Angular momentum: an approach to combinatorial space-time'', in Quantum theory and beyond, T. Bastin ed., Cambridge Univ. Press., New York, 151-180.
  
  \bibitem{Schwinger}
  J.~Schwinger,
  ``On Angular Momentum,''
  U.S. \ Atomic \ Energy \ Commission.\ (unpublished) \ NYO-3071, (1952).
  
\bibitem{Bargmann}
  V.~Bargmann,
  ``On the Representations of the Rotation Group,''
  Rev.\ Mod.\ Phys.\  {\bf 34}, 829 (1962).


\bibitem{Labarthe}
  J.~J.~Labarthe,
  ``Generating Functions For The Coupling Recoupling Coefficients Of SU(2),''
  J.\ Phys.\ A  {\bf 8}, 1543 (1975).

\bibitem{Westbury}
 B.W.~Westbury, 
 ``A generating function for spin network evaluations, Knot theory,''
 Banach \ Center \ Publ.\, {\bf 42}, 447 (1998).

\bibitem{Garoufalidis}
 S.~Garoufalidis and R.~Van der Veen,
 ``Asymptotics of classical spin net-works,''
 arXiv:0902.3113 (2009).

\bibitem{Costantino}
 F.~Costantino and J.~March«e 
 ``Generating series and asymptotics of classical spin networks,''  arXiv:1103.5644 (2011).


  
  
  
   \bibitem{BarrettPR} 
  J.~W.~Barrett and I.~Naish-Guzm
  ``The Ponzano-Regge model,''
  Class.\ Quant.\ Grav.\  {\bf 26}, 155014 (2009)
  [arXiv:0803.3319 [gr-qc]].
  
  \bibitem{LS} 
  E.~R.~Livine and S.~Speziale,
  ``A New spinfoam vertex for quantum gravity,''
  Phys.\ Rev.\ D {\bf 76}, 084028 (2007)
  [arXiv:0705.0674 [gr-qc]].
  \bibitem{holomorph}
L. Freidel, K. Krasnov and E.R. Livine,
{\it Holomorphic Factorization for a Quantum Tetrahedron},
arXiv:0905.3627
    
  \bibitem{livine-bonzom} 
  V.~Bonzom and E.~R.~Livine,
  ``A New Hamiltonian for the Topological BF phase with spinor networks,''
  arXiv:1110.3272 [gr-qc].
  
   \bibitem{LJ} 
  V.~Aquilanti, H.~M.~Haggard, A.~Hedeman, N.~Jeevanjee, R.~G.~Littlejohn and L.~Yu,
  ``Semiclassical Mechanics of the Wigner $6j$-Symbol,''
  arXiv:1009.2811 [math-ph].

 \bibitem{Ljeff}
 L. Freidel and J. Hnybida,
 ``generalised 4-valent intertwiners'',
 to appear.
  
  
  
  
  
  
  \bibitem{Mahajan}
M. Mahajan and V. Vinay,
{\it Determinant: Old Algorithms, New Insights},
SIAM J. Discrete Math. 12 (1999) 474-490


\bibitem{Kovacs}
  I.~Kovacs and D.S.~Silver and S.G.~Williams,
  ``Determinants of Commuting-block matrices,''
  Am. \ Math. \ Mon.\ {\bf 10}, (1999). 

\bibitem{Gunter}
G\"unter Rote,
{\it Division-Free Algorithms for the Determinant and the Pfaffian: Algebraic and Combinatorial Approaches},
Computational Discrete Mathematics, 2122/2001 (2001) 119-135
  
  
  
  
  
   
 
 
%
%
%
%
%
%
%
%
%
%
%
%
%
%
%
%



\end{thebibliography}
\end{document}